\documentclass[twocolumn,british,prl,superscriptaddress]{revtex4-1}
\pdfoutput=1
\usepackage[T1]{fontenc}
\usepackage[latin9]{inputenc}
\usepackage{babel}

\usepackage{amsmath}
\usepackage{amssymb}
\usepackage{url}
\usepackage{mathtools}
\usepackage{graphics,graphicx,color,float}
\usepackage{etoolbox} 
\usepackage{amsfonts}
\usepackage{amsthm, hyperref}
\usepackage[usenames,dvipsnames]{xcolor}
\usepackage{hyperref}
\usepackage{tikz}
\usepackage[T1]{fontenc}
\usetikzlibrary{matrix}
\usepackage{xr-hyper}
\usepackage{tikz}
\usetikzlibrary{calc}

\newcommand{\la}{\langle}
\newcommand{\ra}{\rangle}

\newcommand{\calH}{\mathcal{H}}

\newcommand{\R}{\mathbb{R}}
\newcommand{\mcs}{\mathcal{S}}
\newcommand{\mcp}{\mathcal{P}}

\newcommand{\C}{\mathbb{C}}

\newcommand{\beq}{\begin{equation}}
\newcommand{\enq}{\end{equation}}
\newcommand{\bel}{\begin{lemma}}
\newcommand{\enl}{\end{lemma}}
\newcommand{\bet}{\begin{theorem}}
\newcommand{\ent}{\end{theorem}}

\makeatletter
\newcommand*{\addFileDependency}[1]{% argument=file name and extension
  \typeout{(#1)}
  \@addtofilelist{#1}
  \IfFileExists{#1}{}{\typeout{No file #1.}}
}
\makeatother

\newcommand{\Tr}{\mathrm{tr}}

\newcommand{\eps}{\varepsilon}

\usepackage{booktabs,tabularx}

\newcommand{\suppress}[1]{}

\newcommand{\bra}[1]{\langle #1|}
\newcommand{\ket}[1]{|#1 \rangle}
\newcommand{\braket}[2]{\langle #1|#2\rangle}

\mathchardef\mhyphen="2D

\def\be{\begin{equation}}
\def\ee{\end{equation}}

\newcommand{\gexcl}{\mathcal{G}_{\mathrm{ex}}}

\makeatletter
\newcommand*{\rom}[1]{\expandafter\@slowromancap\romannumeral #1@}
\makeatother

\makeatletter
\appto{\appendix}{%
 \@ifstar{\def\theequation@prefix{A.}}%
 {}%
}
\makeatother

\mathchardef\mhyphen="2D
\newtheorem{theorem}{Theorem}
\newtheorem{lemma}[theorem]{Lemma}

\newtheorem{definition}[theorem]{Definition}
\newtheorem*{main result}{Main Theorem}
\newtheorem*{proof strategy}{Proof Strategy}
\newtheorem*{cor}{Corollary}

\newtheorem*{theorem*}{Theorem}

\usepackage[most]{tcolorbox}
\newtcolorbox{mybox}[2][]{%
  attach boxed title to top center
               = {yshift=6pt},
  colback      = black!1!white,
  colframe     = black!15!black,
  fonttitle    = \bfseries,
  colbacktitle = black!15!black,
  title        = #2,#1,
  enhanced,
}
\begin{document}
	
\title{Robust self-testing of quantum systems via non-contextuality inequalities}

%%%%%%%%%%%%%%%%%%%%%%%%%%%%%%%%%%%%%%%%%%%%%%%%%%%%%%%%%%%%%%%%%%%

\begin{abstract}
Characterising 
%Self-testing 
unknown quantum states and measurements is a fundamental problem in quantum information processing. In this Letter, 
we provide a novel scheme to self-test local quantum systems  using  non-contextuality inequalities. Our work  leverages  the graph-theoretic framework for contextuality introduced by Cabello, Severini, and Winter, combined with  tools  from mathematical optimisation that guarantee  the unicity  of optimal solutions. 
As an application, 
%Furthermore,
 we show that the celebrated Klyachko-Can-Binicio\ifmmode \breve{g}\else \u{g}\fi{}lu-Shumovsky inequality  and its generalisation to contextuality scenarios with odd $n$-cycle compatibility relations 
%graphs 
admit robust  self-testing. 
%We show that these results are robust subject to the appropriate compatibility assumptions, demonstrating the use of contextuality for self-testing applications. 
%We arrive at our results using %Our result opens many interesting questions for future research.
\end{abstract}

\author{Kishor Bharti}
\affiliation{Centre for Quantum Technologies, National University of Singapore}
%\email{e0016779@u.nus.edu}

\author{Maharshi Ray}
\affiliation{Centre for Quantum Technologies, National University of Singapore}
%\email{maharshi91@gmail.com}

\author{Antonios Varvitsiotis}
\affiliation{Department of Electrical and Computer Engineering, National University of Singapore}
%\email{avarvits@gmail.com}

\author{Naqueeb Ahmad Warsi}
\affiliation{Centre for Quantum Technologies, National University of Singapore}
%\email{warsi.naqueeb@gmail.com}

\author{Ad\'{a}n Cabello}
\affiliation{Departamento de F\'{i}sica Aplicada II, Universidad de Sevilla, E-41012 Sevilla, Spain}
%\email{adan@us.es}

\author{Leong-Chuan Kwek}
\affiliation{Centre for Quantum Technologies, National University of Singapore} \affiliation{MajuLab, CNRS-UNS-NUS-NTU International Joint Research Unit, Singapore UMI 3654, Singapore}
%\affiliation{Institute of Advanced Studies, Nanyang Technological University, Singapore 639673, Singapore}
\affiliation{National Institute of Education, Nanyang Technological University, Singapore 637616, Singapore}
%\email{cqtklc@gmail.com}

\maketitle

%%%%%%%%%%%%%%%%%%%%%%%%%%%%%%%%%%%%%%%%%%%%%%%%%%%%%%%%%%%%%%%%%%%

%\tableofcontents

%%%%%%%%%%%%%%%%%%%%%%%%%%%%%%%%%%%%%%%%%%%%%%%%%%%%%%%%%%%%%%%%%%%

%\section{Introduction}
\noindent {\em Introduction.---}The deployment and analysis of mathematical models have been a crucial tool to advance our scientific understanding of the physical world. Nevertheless, complex mathematical models often admit a multitude of possible solutions, a phenomenon that can lead to ambiguity and erroneous predictions when the solution of the model is used to study some real-life problem. Models with no uniquely-identifiable solutions manifest themselves across most fields of science and mathematics, typical examples being the nonuniqueness of solutions to partial differential equations and the existence of multiple Nash equilibria in non-coopearative games. More pertinent to this work, 
the uniqueness of the ground state of a Hamiltonian is a problem with important engineering applications. Indeed, quantum annealing crucially relies on the uniqueness of the ground state of the underlying Hamiltonian, which is used to encode the solution of an optimization problem \cite{das2005quantum}.

From a practical standpoint, the noisy nature of the collected data governing the model selection process, suggests we should employ ``robust'' models, i.e., models that have a unique solution that is moreover stable under perturbations of the input data. Notwithstanding the ubiquitousness and importance of problems related to the unicity and robustness of the solutions of a given model, there is no general framework allowing to address these questions in a unified manner.

One of the most extensively used modeling tools in science and engineering is mathematical optimization. In this setting, the model is specified by a family of decision variables that satisfy certain feasibility constraints. The goal is then to find the value of the decision variables that maximizes an appropriate measure of performance. Undoubtedly, one of the most important optimization model is linear programming (LP), where the decision variables are scalar variables subject to affine constraints. An equally important optimization model is semidefinite programming~(SDP), constituting a wide generalization of linear programming with extensive modeling power and efficient algorithms for solving them. Unlike linear programs, the decision variables in a SDP are vectors, and the constraints are defined in terms of the inner products of the vectors. SDPs have many important applications in physics, e.g. in quantum foundations (Bell nonlocality, contextuality, steering) \cite{tura2014detecting, CSW, cavalcanti2016quantum}, quantum information theory (entanglement witnesses, tomography, quantum state discrimination) \cite{doherty2002distinguishing, gonccalves2013quantum, jevzek2002finding}, quantum cryptography \cite{masanes2011secure}, and quantum complexity \cite{reichardt2009span}, just to mention a few. Most importantly, the aspect of SDPs that is crucial to this work is that, like LPs they offer a general framework for studying uniqueness and robustness of model solutions.

In this Letter, we employ the paradigm of identifiable robust models to characterize untrusted devices via contextuality. 
Contextuality refers to the impossibility of reproducing a set of probability distributions, each of them for a context (defined as a set of compatible and mutually nondisturbing observables), that share some marginal probabilities with a joint probability distribution in a single probability space. Quantum theory is an example of a contextual theory \cite{KS67}. In this work we appropriately extend the paradigm of Bell self-testing to the framework of contextuality.
In terms of techniques, our work leverages the well-known link between contextuality and semidefinite programming identified in the seminal work by Cabello, Severini, and Winter \cite{CSW}, combined with some less-known results concerning the unicity and robustness of optimal solutions to semidefinite programs. Roughly speaking, we show that the nearness-of-optimality of the CSW semidefinite program bounds the distance in the SDP-solution space, which in turn translates into a bound on the distance from the ideal quantum realization. We believe that the tools employed in this paper will have value outside of the domain of contextuality, e.g., see \cite{thinh2018structure} for a recent application in Bell nonlocality. Our results render new insights into the foundations of quantum contextuality and a proof-of-principle approach to characterize the underlying quantum states and measurements manifesting quantum contextuality via experimental statistics. 
We provide an innovative scheme to attest robust self-testing for any noncontextuality inequality and present a concrete illustration for the case of 
the generalized KCBS inequality, which is defined for any odd number of measurement events $n\ge 5$. 
Lastly, in terms of applications, our results allows one to verify quantum systems locally under the following three assumptions characteristic of Kochen-Specker contextuality scenarios \cite{KS67,KCBS,Cabello08}. {Assumption 1:} The measurements are {\em ideal} \cite{CY14,CY16} (i.e., they give the same outcome when performed consecutive times, they do not disturb compatible measurements, all their coarse-grainings admit realizations that satisfy these properties), {Assumption 2:} The measured system has no more memory than its information carrying capacity, each measurement device is only used once, and there is an unlimited supply of them. {Assumption 3:} The measurements obey the compatibility relations dictated by the odd cycle graph. 
In the case of Bell self-testing, it is necessary to assume that the involved parties are spacelike separated and that there is no superluminal communication \cite{Bell64, CHSH}, for otherwise, the statistics that attain the quantum supremum of a Bell inequality \cite{toner2003communication} can be simulated using classical resources. In the same spirit, Assumption~2 is necessary in the setting of contextuality, for otherwise, contextuality can be simulated by classical systems \cite{kleinmann2011memory, Cabello2018optimal}. 

%%%%%%%%%%%%%%%%%%%%%%%%%%%%%%%%%%%%%%%%%%%%%%%%%%%%%%%%%%%%%%%%%%%

\medskip 
%\section{Main result}\label{sec:main}
\noindent {\em Self-testing in Bell scenarios.---}To motivate our results, it is instructive to survey the relevant results in the setting of Bell nonlocality, a special case of contextuality where the contexts are generated by the spacelike separation of the involved parties~\cite{amaral2018graph}. 
The experimental tests which reveal the nonlocal nature of a physical theory are called Bell inequalities or Bell tests. 
Geometrically, a Bell inequality corresponds to a halfspace that contains the set of local behaviours,~i.e., 
\begin{equation}\label{csver}
\sum_{a,b,x,y}B^{ab}_{xy}p(ab|xy)\le B_{\ell}, 
\end{equation}
for all local behaviors $p(ab|xy).$ 
The quantum supremum of the Bell inequality \eqref{csver}, denoted by $B_q$, is the largest possible value of the expression $\sum_{a,b,x,y}B^{ab}_{xy}p(ab|xy),$ when $p(ab|xy)$ ranges over the set of quantum bahaviors, i.e., 
$$p(ab|xy)=\bra{\psi} A_{x|a}\otimes B_{y|b}\ket{\psi},$$
for a quantum state $\ket{\psi}\in \calH_A \otimes \calH_B$ and quantum measurements $\{A_{x|a}\}$,~$\{B_{y|b}\}$ acting on $\calH_A$ and $\calH_B$ respectively. 
Besides their physical significance, Bell inequality violations witness the existence of certifiable
randomness, and have been leveraged to power many other important information-theoretic tasks~\cite{ekert1991quantum, colbeck2009quantum,cubitt2010improving}.

The feature of Bell inequalities that is most pertinent to this work is that, the quantum realizations that achieve the quantum supremum of a Bell inequality are sometimes uniquely determined up to local isometries and ancilla degrees of freedom. Formally, a Bell-inequality is a self-test for the realization $(\calH_A, \calH_B, \psi, \{A_{x|a}\}, \{B_{y|b}\})$ if for any other realization $(\calH_{A'}, \calH_{B'}, {\psi}', \{A'_{x|a}\}, \{B'_{y|b}\})$ that also attains the quantum supremum, there exists a local isometry $V=V_A\otimes V_B$ and an ancilla state $\ket{junk}$ such~that 
\be\label{bell:ST}
\begin{aligned}
V\ket{\psi'}&=\ket{junk}\otimes \ket{\psi} ,\\
V(A'_{x|a}\otimes B'_{y|b})\ket{\psi'} & =\ket{junk} (A_{x|a}\otimes B_{y|b})\ket{\psi}.
%B'_{y|b} &=V_B(B_{y|b}\otimes 1_{\calH_{B''}})V_B^\dagger.
\end{aligned}
\ee
In practical terms however, when a Bell experiment is performed in the lab, experimental imperfections will only allow to achieve a value which is close, but not equal, to the ideal quantum supremum $B_q$. These practical considerations naturally lead to the notion of robust self-testing. Specifically, a Bell inequality is an $(\epsilon,r)$-robust self-test for the realization $(\calH_A, \calH_B, \psi, \{A_{x|a}\}, \{B_{y|b}\})$ if it is a self-test as defined above, and furthermore, whenever for some realization $(\calH_{A'}, \calH_{B'}, {\psi}', \{A'_{x|a}\}, \{B'_{y|b}\}),$ 
$$\sum_{a,b,x,y}B^{ab}_{xy}\bra{\psi'} (A'_{x|a}\otimes B'_{y|b})\ket{\psi'}\ge B_q-\epsilon,$$
we have that 
$$\|V(A'_{x|a}\otimes B'_{y|b})\ket{\psi'} -\ket{junk} (A_{x|a}\otimes B_{y|b})\ket{\psi}
\|\le \mathcal{O}(\epsilon^r).$$
As an example, the well-known Clauser-Horne-Shimony-Holt (CHSH) Bell inequality is an $(\epsilon, {1\over 2})$ robust self-test for the singlet state and appropopriate Pauli measurements, e.g. see \cite{popescu1992generic,Yao_self,mckague2012robust,bugliesi2006automata}. The term ``self-testing'' was first introduced by Mayers and Yao \cite{Yao_self} in the setting of Bell-nonlocality \cite{Bell64}. However, the idea underlying self-testing is present in earlier works, for example in the works of Tsirelson \cite{tsirel1987quantum}, Summers-Werner \cite{summers1987bell}, and Popescu-Rohrlich \cite{popescu1992generic}. Recent research on self-testing moves in various new directions, e.g. which states can be self-tested \cite{ vsupic2018self, Koon_Tong} or how to tighten the robustness results, so that self-testing results have practical applications \cite{Jed}.

%%%%%%%%%%%%%%%%%%%%%%%%%%%%%%%%%%%%%%%%%%%%%%%%%%%%%%%%%%%%%%%%%%%

\medskip 
%\section{Main result}\label{sec:main}
\noindent {\em Self-testing in contextuality scenarios.---} In this section we introduce a natural analogue of the notion of (robust) self-testing for contextuality scenarios, where the noncontextuality assumption is not enforced via locality. We follow the exclusivity graph approach to contextuality \cite{CSW}.

A contextuality scenario is defined by a family of measurement events $e_1,\ldots,e_n$. Two events are {\em mutually exclusive} when they can be realized by the same measurement but correspond to different outcomes. To the events $\{e_i\}_{i=1}^n$ we associate their {exclusivity graph}, whose vertex set is $\{1,\ldots,n\}$ (denoted by $[n]$), and two vertices $i,j$ are adjacent (denoted by $i\sim j$) if the measurement events $e_i$ and $e_j$ are exclusive. 

For an exclusivity graph $\gexcl$, we consider theories that assign probabilities to the measurement events corresponding to its vertices. A {behavior} corresponding to $\gexcl$ is a mapping $p: [n]\to [0,1]$, where $p_i+p_j\le 1$, for all $i\sim j$. 
Here, the nonnegative scalar $p_i\in [0,1]$ encodes the probability that measurement event $e_i$ occurs. Furthermore, note that the linear constraint $p_i+p_j\le 1$ enforces that if measurement event $e_i$ takes place (i.e., $p(e_i)=~1)$, the event $e_{i+1}$ cannot take~place.

A behavior $p: [n]\to [0,1]$ is deterministic noncontextual if all events have pre-determined values that do not depend on the occurrence of other events. Concretely, a deterministic noncontextual behavior $p$ is a mapping $p: [n]\to \{0,1\},$ where $p_i+p_j\le 1$, for all $i\sim j$. The polytope of {noncontextual behaviors}, denoted by $\mathcal{P}_{nc}(\gexcl)$, is the convex hull of all deterministic noncontextual behaviors. Behaviors that do not lie in $\mathcal{P}_{nc}(\gexcl)$ are {contextual}. A behavior $p: [n]\to~[0,1]$ is {quantum} if there exists a quantum state $\rho$ and projectors $\Pi_1,\ldots \Pi_n$ acting on a Hilbert space $\mathcal{H}$~where 
\be\label{cwfere}
p_i=\Tr(\rho\Pi_i), \forall i\in [n] \text{ and } \Tr(\Pi_i\Pi_j)=0, \text{ for } i\sim j.
\ee 
We refer to the realization $\rho, \{\Pi\}_{i=1}^n$ satisfying \eqref{cwfere} as a quantum realization of the behavior $p$. 
The convex set of all quantum behaviors is denoted by $\mathcal{P}_{q}(\gexcl)$. For the purposes of this manuscript, we will denote a quantum realization by $\{\ket{u_i}\bra{u_i}\}_{i=0}^n$. Furthermore, $p_i \equiv \vert \bra{u_o} u_i \rangle \vert ^2 , \forall i\in [n] \text{ and } p \in \mathcal{P}_{q}(\gexcl)$. A noncontextuality inequality corresponds to a halfspace that contains the set of noncontextual behaviours, i.e., 
\beq\label{eq:ncineq}
 \sum_{i \in [n]} w_i p_i \leq B_{nc}(\gexcl, w), 
\enq
for all $p \in \mathcal{P}_{nc}(\gexcl),$ where $w_1,\ldots,w_n\ge 0$.% and $B_{nc}$ are real numbers.
%Here, $B_{nc} $ is the noncontextual hidden variable bound on the linear expression $S_w$ and is obtained by taking behaviours over $B_{nc}(\gexcl)$. 
The quantum supremum of the noncontextuality inequality~\eqref{eq:ncineq}, denoted by $B_{qc}(\gexcl, w)$, is the largest value of the expression $\sum_{i \in [n]} w_i p_i$, as $p$ ranges over the set of quantum behaviors $\mathcal{P}_{q}(\gexcl)$.
%However, it is possible to violate this bound by taking behaviours over $B_q(\gexcl)$ and we will represent the maximum achievable bound as $B_{cq}.$

%As a first step, we introduce the appropriate notion of uniqueness.
Motivated by Bell self-testing, we now introduce a natural notion of ``uniqueness'' for the quantum realizations $\{\ket{u_i}\bra{u_i}\}_{i=0}^n$ that attain the quantum supremum of a noncontextuality inequality. In this setting, uniqueness refers to identifying the state and measurement operators that 
achieve the quantum supremum, up to a {\em global isometry}. 
The notion of uniqueness and robustness appropriate for our work is introduced below. 
\begin{definition} (\textbf{Self-testing})
 A noncontextuality inequality  $\sum_{i \in [n]} w_i p_i \leq B_{nc}(\gexcl, w)$ is a self-test for the realization $\{\ket{u_i}\bra{u_i}\}_{i=0}^n$~if:
\begin{enumerate}
\item $\{\ket{u_i}\bra{u_i}\}_{i=0}^n$ achieves the quantum supremum $B_{qc}(\gexcl, w)$;
\item For any other realization $\{\ket{u_i'}\bra{u_i'}\}_{i=0}^n$ that also achieves $B_{qc}(\gexcl, w)$, there exists an isometry $V$ such that 
%We say that Given a behaviour $p\in B_Q(\gexcl)$ such that $p$ achieves the maximum quantum bound (i.e $B_{qc}$) on a noncontextuality inequality defined for $\gexcl$ and $\{\ket{u_i}\}_{i=0}^n$ be a quantum realization of $p$. %That is, $p_i=|\braket{u_0}{u_i}|^2 $ satisfying $\braket{u_i}{ u_j}=0$ whenever $(i, j) \in E(\gexcl)$. 
\beq\label{eq:st}
V\ket{u_i}\bra{u_i}V^{\dagger}=\ket{u_i'}\bra{u_i'}, \quad 0\le i\le n.
\enq
\end{enumerate}
\end{definition}
%Let $\sum_{i \in [n]}w_i p_i , w_i \geq 0 $ be the linear expression corresponding to a noncontextuality inequality for a scenario represented by exclusivity graph $\gexcl$ with noncontextual hidden variable bound $B_{nc}$ and quantum bound $B_{cq}.$ Let $p \in B_Q(\gexcl)$ be a self-test for the realization $\{\ket{u_i}\}_{i=0}^n$ 
%manifested via the same noncontextuality inequality. 
\begin{definition}(\textbf{Robustness}) A noncontextuality inequality $\sum_{i \in [n]} w_i p_i \leq B_{nc}(\gexcl, w)$ is an $(\epsilon, r)$-robust self-test for $\{\ket{u_i}\bra{u_i}\}_{i=0}^n$ if it is a self-test, and furthermore, for any other realization $\{\ket{u_i'}\bra{u_i'}\}_{i=0}^n$ satisfying 
$$\sum_{i=1}^n w_i |\braket{u_i'}{u'_0}|^2\ge B_{qc}(\gexcl, w)-\epsilon,$$
there exists an isometry $V$ such~that 
\be\label{eq:rst}
\|V\ket{{u_i}}\bra{{u_i}}V^{\dagger} - \ket{{u'_i}}\bra{{u'_i}}\| \leq \mathcal{O}\left({\epsilon}^r\right), \quad 0\le i\le n. \ee 
\end{definition}
%
%We focus on a particular noncontextuality inequality, namely the Klyachko-Can-Binicio\ifmmode \breve{g}\else \u{g}\fi{}lu-Shumovsky (KCBS) inequality, which bounds the sum of the probabilities of the occurrence of five individual events
%%where there are five measurement events
% that are cyclically exclusive \cite{KCBS}. The % Choosing as our figure of merit the sum of the probabilities of the occurrence of the five individual events, 
% largest value of the KCBS inequality within any noncontextual hidden variable (NCHV) theory is equal to 2~\cite{add}, whereas quantum theory violates the NCHV bound, thus establishing the contextual nature of quantum theory \cite{add}. Specifically, it is well-known that
%\be\label{kcbs5}
%\sum_{i=1}^5 \Tr(\rho \Pi_i) \le\sqrt{5},
%\ee
% for all quantum states $\rho$ and projectors $\Pi_1,\ldots \Pi_5$ acting on a single Hilbert space $\mathcal{H}$ that satisfy $\Tr(\Pi_i\Pi_{i+1})=~0$. Moreover, the quantum supremum $\sqrt{5}$ can be attained by a qutrit state, a fact that also identifies the smallest single indivisible quantum system that witnesses contextuality.
% %Analogously 
%The formal definition is given in \eqref{eq:st} and~\eqref{eq:rst}. 
This definition of self-testing is in stark contrast to the case of Bell self-testing, where uniqueness is defined up to {\em local isometries} (recall \eqref{bell:ST}) to account for
 the physical operational freedom of spacelike separated parties to pre-process
their local quantum systems and measurements.
 %the tensor product structure of the underlying Hilbert~space. 
 Furthermore, unlike the case of Bell self-testing, in contextuality scenarios there is no meaningful sense in
which the state can be self-tested in isolation, rather, a state is always self-tested in relation to a measurement.
On a side note, it is worth noticing that closeness between a pair of quantum  \emph{realizations} implies closeness of the corresponding pair of  \emph{quantum} behaviors. 

%%%%%%%%%%%%%%%%%%%%%%%%%%%%%%%%%%%%%%%%%%%%%%%%%%%%%%%%%%%%%%%%%%%

\medskip 
%\section{Main result}\label{sec:main}
\noindent {\em How to show self-testing.---}To show that a noncontextuality inequality is a self-test we rely on the connection with SDPs established in \cite{CSW}, where it was shown that the quantum supremum of a noncontextuality inequality
 \be\label{cdfvdfb}
 \max \left\{ \sum_{i=1}^{n} w_ip_i : p \in \mathcal{P}_{q}(\gexcl) \right \},
 \ee
 is equal to the value of the following SDP: 
 \be\label{theta:primalmain}
 \begin{aligned} 
\vartheta(\gexcl,w) = \max & \ \ \sum_{i=1}^n w_i{X}_{ii} \\
\text{ subject to} & \ \ {X}_{ii}={ X}_{0i}, \ \ 1\le i\le n,\\
 & \ \ { X}_{ij}=0, \ \ i\sim j,\\
& \ \ X_{00}=1,\ \ X\in \mathcal{S}^{1+n}_+,
\end{aligned}
\ee 
 % 
% \be\label{cnintro}
%\begin{aligned} 
%\max & \ \sum_{i=1}^n { X}_{ii} \\
%\text{ subject to } & \ { X}_{ii}={ X}_{0i}, \quad 1\le i \le n,\\
% & \ { X}_{i(i+1)}=0,\quad 1\le i\le n,\\
%& \ X_{00}=1,\ X\in \mathcal{S}^{1+n}_+,
%\end{aligned}
%\ee
where $\mathcal{S}^{1+n}_+$ denotes the cone of positive semidefinite matrices of size $n+1$. The optimization program \eqref{theta:primalmain} is known as the Lov\'asz theta number of the vertex-weighted graph $(\gexcl,w)$ \cite{lovasz1979shannon} where vertex weighted graph refers to a graph where a weight is assigned to each vertex.
Moreover, the equivalence between the optimization problems \eqref{cdfvdfb} and \eqref{theta:primalmain} also induces a correspondence between their optimal solutions. \emph{Specifically, if $p \in \mathcal{P}_{q}(\gexcl)$ is optimal for \eqref{cdfvdfb} and $ \{\ket{u_i}\bra{u_i}\}_{i=0}^n$ is a quantum realization of $p$, the Gram matrix of the vectors $\ket{u_0},\braket{u_0}{u_1}\ket{u_1},\ldots,\braket{u_0}{u_n}\ket{u_n}$ corresponds to an optimal solution for~\eqref{theta:primalmain}. 
Conversely, for any optimal solution $X={\rm Gram}(\ket{u_0},\ket{u_1},\ldots,\ket{u_n})$ of the SDP \eqref{theta:primalmain}, %if $u_0, u_1, \ldots, u_n$ is an arbitrary Gram decomposition of $X$, 
the realization $\{{\ket{u_i}\bra{u_i} \|\ket{u_i}\bra{u_i}\|^{-1}}\}_{i=0}^n$ is optimal for \eqref{cdfvdfb}.}
This correspondence leads to the following three-step proof strategy for showing that the noncontextuality inequality \eqref{cdfvdfb} is an $(\epsilon, {1\over 2})$-robust self-test.
%\begin{mybox}[colback = white, width = 8.7cm]{Proof Strategy}
\begin{itemize}
 \item First, show that the SDP \eqref{theta:primalmain} has a unique optimal solution $X^*$. 
 \item Second, show that any $\epsilon$-suboptimal solution $X$ of \eqref{theta:primalmain}, i.e. a feasible $X$ where $\sum_iw_iX_{ii}\ge \vartheta(\gexcl)-\epsilon$, satisfies $\|\tilde{X}-X^*\|_F\le \mathcal{O}(\epsilon)$. 
\item Third, show that for two positive semidefinite matrices that are $\epsilon$-close in Frobenius distance, the vectors in their Gram decompositions are $\mathcal{O}\left(\sqrt{\epsilon}\right)$ close in $2$-norm. % On the other hand, we were unable to show that \eqref{theta:primalmain} admits a unique solution for arbitrary exclusivity graphs, but instead, we identified a sufficient condition guaranteeing this to be the~case. 
%As a last step, we rely on the general fact that for
 %two positive semidefinite matrices that are $\epsilon$-close in Frobenius distance, their Gram decompositions are $\mathcal{O}\left(\sqrt{\epsilon}\right)$ close, cf. Theorem~\ref{equivalentapprox} in the Supplementary material (SM)
 \end{itemize}

%\end{mybox}
Whenever the first step holds, the second step is satisfied for the SDP \eqref{theta:primalmain} and the third step is always true. 
 The proofs of the first two steps   hinge on the rich duality theory of SDPs. Specifically, 
the Lagrange dual of the
 SDP~\eqref{theta:primalmain}, is given by the least scalar $t\ge 0$ for which 
\be \label{dualcnintro}
 %tE_{00}+\sum_{i=1}^n \lambda_iE_{0i}+\sum_{i=1}^n(\nu_i-\lambda_i)E_{ii}+\sum_{i\sim j} \mu_{ij}E_{ij}\succeq 0,
 Z \equiv  tE_{00}+\sum_{i=1}^n (\lambda_i-w_i)E_{ii}-\sum_{i=1}^n\lambda_i E_{0i}+\sum_{i\sim j} \mu_{ij}E_{ij}  \succeq 0,
\ee where $E_{ij} = \frac{e_ie_j^{\top} + e_je_i^{\top}}{2}$ and the column vectors  $\{e_i\}_{i= 0}^n$ form the standard basis of $\R^{n+1}$. Furthermore, $\lambda_i$, $w_i$ and $\mu_{ij}$ are the Lagrange multipliers corresponding to the constraints of the primal SDP  \eqref{theta:primalmain}. %Call this dual variable matrix as $Z$. 
 The first tool  we use is a sufficient condition for showing that an arbitrary  SDP admits a unique optimal solution, in terms of the existence of an appropriate   optimal solution for its dual problem, see Theorem~\ref{uniqueness} in the Appendix. %  that satisfies certain linear conditions. Specifically, see Theorem \ref{uniqueness} in SM.
The second crucial tool  are  error bounds for SDPs, which allow to  bound  the distance of a feasible solution from the set  of optimal solutions, in terms of  the suboptimality of the objective function, see  Theorem~\ref{thm:errorbounds} in the Appendix. Combining these  two tools, we arrive at   our main technical tool, allowing  to show that a non-contextuality inequality is a~self-test:

%first derived in \cite{sturm} (see also~\cite{wolk} for a modern exposition). Given two sets $Q_1, Q_2$ lying in some Euclidean space, we say that they satisfy a {\em H\"older error bound} with exponent $q\ge 0$, if for any compact set $U$ we have that 
%\be{\rm dist}(x,Q_1\cap Q_2)=\mathcal{O}({\rm dist}^q(x,Q_1)+{\rm dist}^q(x,Q_2)), \ \forall x\in U.   
%\ee
%Recall that for a set $Q\subseteq \R^n$, the corresponding distance function is given by ${\rm dist}(x,Q)=\inf\{ \|x- q\|: q\in Q\}.$
%Clearly, if $Q$ is closed, there always exists a point that achieves the minimum distance. If additionally $Q$ is convex, the distance is minimized at a unique point of $Q$. 
%As it turns out, {\em feasibility SDPs} satisfy H\"olderian error bounds, where the H\"older exponent depends neither on the size of the matrices $n$ nor the number of affine constraints~$m$, see Theorem \ref{thm:errorboundsfeas} in SM. Based on this, one can prove error bounds for arbitrary SDPs, see Theorem \ref{thm:errorbounds} in SM.

\begin{main result}  \label{res:uniqueness} Consider a  noncontextuality inequality 
$\sum_{i=1}^{n} w_ip_i \le B_{nc}(\gexcl, w)$. Assume that  
\begin{enumerate}
\item There exists an optimal quantum realization $\{\ket{u_i}\bra{u_i}\}_{i=0}^n$ such that  $$ \sum_i w_i \vert \braket{u_i}{u_0}\vert^2 = B_{qc}(\gexcl, w)$$ and $\braket{u_0}{u_i}\ne 0, $ for all $1\le i\le n$, and
\item There exists   a dual optimal solution $Z^*$ for the SDP \eqref{dualcnintro} such  that the homogeneous linear system 
\be \label{algcon}
\begin{aligned}
M_{0,i}  &= M_{i,i},  \text{ for all } 1 \leq i \leq n, \\
M_{i,j} &= 0,  \text{ for all } i \sim j, \\
MZ^*&= 0,
\end{aligned}
\ee
in the symmetric matrix variable $M$ only admits the trivial solution $M=0$.
\end{enumerate}
Then,  the noncontextuality inequality is an $(\epsilon, {1\over 2})$-robust self-test for $\{\ket{u_i}\bra{u_i}\}_{i=0}^n$.
\end{main result}
%the primal SDP \eqref{primalintro} has a unique optimal solution. 

The proof of the main  theorem is  deferred to Section~\ref{mainresultproof} in the Appendix. In the next section we shift our focus to a particular instance of the main theorem, namely the KCBS non-contextuality inequalities. Specifically, we show that for the KCBS  inequalities, condition \eqref{algcon} is satisfied and therefore, such inequalities are robust~self-tests.

\medskip 
\noindent {\em An application:~The KCBS inequalities.---}A celebrated noncontextuality inequality is the Klyachko-Can-Binicio\u{g}lu-Shumovsky (KCBS) inequality, first introduced for $n=5$ in \cite{KCBS} and subsequently generalized to general odd values of $n$ \cite{badziag2011pentagrams,Liang}. The KCBS inequality corresponds to an odd number of measurement events $e_1,\ldots, e_n$ with the property that $e_i$ and $e_{i+1}$ are exclusive, where indices are taken modulo $n$. The corresponding exclusivity graph is the $n$-cycle and the set of noncontextual behaviors is $\mathcal{P}_{nc}(C_n)$. 
Concretely, for any odd $n$, the ${\rm KCBS}_n$ noncontextuality inequality is given~by
\beq
\max\left\{ \sum_{i=1}^{n} p_i: p\in \mathcal{P}_{nc}(C_n)\right\}={(n-1)\over 2}. \\
\enq 
The ${\rm KCBS}_n$ inequality witnesses quantum contextuality, as quantum behaviors can achieve values greater then ${(n-1)/ 2}$. Specifically, for any odd $n$, we have that 
\be\label{q:valuekcbs}
 \max \left\{ \sum_{i=1}^{n} p_i : p \in \mathcal{P}_{nc}(C_n) \right \} = \frac{ n \cos \pi/ n}{ 1+ \cos\pi / n},
\ee 
and a quantum behavior in $\mathcal{P}_{q}(C_n)$ that achieves the quantum supremum of the ${\rm KCBS}_n$ inequality is: 
\be\label{behavpn}
 p^{(n)}_i=\frac{ \cos \pi/ n}{ 1+ \cos\pi / n}, \ 1\le i\le n.
\ee

A quantum realization $ $ which achieves the quantum supremum corresponds to  
\be \label{realization}
\begin{aligned}
\ket{u_0} = \left(1,0,0\right)^{T} \text{ and} \\
\ket{u_j} = \left( \cos(\theta), \sin(\theta)\sin\left( \phi_j\right), \sin(\theta)\cos\left(\phi_j\right) \right)^{T},
\end{aligned}
\ee
 where 
 $\cos^2(\theta)= \frac{ \cos\left(\pi/ n\right)}{ 1+ \cos\left(\pi / n\right)} \text{ and } \phi_j = \frac{j \pi (n-1)}{n} \text{ for } 1\le j\le n.$
As it turns out, the generalised KCBS inequality satisfies the assumptions of the Main Theorem, 
and the resulting self-testing statement is formally stated as~follows:

\begin{cor}\label{maincorollary}
 For any odd integer $n$, the $KCBS_n$ inequality is an $(\epsilon,{1\over 2})$-robust self-test for the realization corresponding to equation \eqref{realization}.%$\{\ket{v_i}\}_{i=1}^0$.}
 %\be\label{KCBSrealization}
 %\begin{aligned}
%\ket{v_j} &= \left( \cos(\theta), \sin(\theta)\sin\left(\frac{j \pi (n-1)}{n}\right), \sin(\theta)\cos\left(\frac{j \pi (n-1)}
%{n}\right) \right)^\top,
%\end{aligned}
%\ee
\end{cor}
By the Main Theorem, the proof of the corollary boils down to finding a dual optimal solution satisfying~\eqref{algcon}. 
In Section~\ref{oddcycleunique} of the Appendix, we show that for any odd integer~$n$, the following matrix has the desired properties:
 \be\label{zmatrix}
{ Z}^*_n = \left[ 
\begin{array} {c | c}
\vartheta(C_n) & -e_n^\top\\
\hline 
-e_n & I_n+\frac{n -\vartheta(C_n)}{2 \vartheta(C_n)}A_{C_n}
\end{array}
\right]\in \R^{(1+n)\times (1+n)},
\ee
where $e_n$ is the all-ones column vector of length $n$, and $A_{C_n}$ is the adjacency matrix of the cycle graph $C_n$. 
 
 Concretely, for $n=5,$ the dual optimal solution  is 
%
%%Solving the dual for 5-cycle we get the following dual optimal solution: 
%%$$t=\sqrt{5}, \quad \mu_i=-2, \quad \lambda_{ij}={5-\sqrt{5}\over 2\sqrt{5}}, i.e., $$
\be\label{dualsolution}
Z_5^{\star}=\left(\begin{array}{c|ccccc} 
\sqrt{5} & -1 & -1 & -1 & -1 & - 1\\
\hline
-1 & 1 & c& 0 & 0 & c\\
-1 & c & 1 & c& 0 & 0 \\
-1 & 0 & c & 1 & c & 0\\
-1 & 0 & 0 & c & 1 & c\\
-1& c & 0 & 0& c & 1
\end{array}\right), 
\ee
where $c={5-\sqrt{5}\over 2\sqrt{5}}.$ 
%The corresponding $M$ matrix (introduced in equation \eqref{algcon}) is given by
Robust self-testing for  the five cycle KCBS inequality corresponds to showing that 
%$M_5 = 0 $ is
 the only solution to the  linear system  $M_5 Z_5^{\star} = 0$ of the form 
 \be\label{x_5}
M_5=\left(\begin{array}{c|ccccc} 
0 & m_1 & m_2 & m_3 & m_4 & m_5\\
\hline
m_1 & m_1  & 0&  m_6& m_9 & 0\\
m_2 & 0 & m_2  & 0& m_7 & m_{10} \\
m_3 & m_6& 0 & m_3  & 0& m_8\\
m_4 & m_9 & m_7 & 0& m_4 & 0\\
m_5 & 0 & m_{10} & m_8& 0 & m_5 
\end{array}\right),
\ee
is the matrix of all-zeros. 

\medskip 
\noindent {\em Conclusions.---}In this work we introduced an appropriate extension of the notion of Bell self-testing to 
the framework of contextuality, where the noncontextuality assumption is not enforced via locality. In our main technical result, we identified a sufficient condition for showing that an arbitrary noncontextuality inequality is a robust self-test. As an application of our main theorem, we showed that the celebrated KCBS noncontextuality inequalities are robust self-tests. Our main theorem is not restricted to KCBS inequalities and can be used to self-test other noncontextuality inequalities, given they satisfy the necessary conditions; this will the be the topic of future investigations. Equally important, our proof techniques leverage a largely unnoticed connection between unicity problems in physics with uniqueness properties of optimization problems, which we believe will be of independent interest to the physics community. 

\medskip 

%%%%%%%%%%%%%%%%%%%%%%%%%%%%%%%%%%%%%%%%%%%%%%%%%%%%%%%%%%%%%%%%%%%

\noindent {\em Acknowledgements.---}We thank Anurag Anshu for helpful discussions concerning Lemma~\ref{equivalentapprox} in the Appendix. We also thank Le Phuc Thinh, Srijita Kundu, and Cai Yu for helpful discussions. We thank the National Research Foundation of Singapore, the Ministry of Education of Singapore, MINECO Project No.~FIS2017-89609-P with FEDER funds, the Conserjer\'{\i}a de Conocimiento, Investigaci\'on y Universidad, Junta de Andaluc\'{\i}a and European Regional Development Fund (ERDF) Grant No.~SOMM17/6105/UGR, and the Knut and Alice Wallenberg Foundation for financial support. AV is supported by the NUS Young Investigator award R-266-000-111-133 and by an NRF Fellowship (NRF-NRFF2018-01 \& R-263-000-D02-281).

\bibliography{ST.bib}
\bibliographystyle{apsrev4-1}
%%%%%%%%%%%%%%%%%%%%%%%%%%%%%%%%%%%%%%%%%%%%%%%%%%%%%%%%%%%%%%%%%%%%
%
\appendix
%
%%%%%%%%%%%%%%%%%%%%%%%%%%%%%%%%%%%%%%%%%%%%%%%%%%%%%%%%%%%%%%%%%%%%
\section{Semidefinite Programming}\label{appsdp}
In this section we briefly review the  most relevant results from  the theory of semidefinite programming, which  we use in the proof of  our main theorem.
%%%%%%%%%%%%%%%%%%%%%%%%%%%%%%%%%%%%%%%%%%%%%%%%%%%%%%%%%%%%%%%%%%%

A pair of primal/dual SDPs is given by
\begin{align}
& \underset{{ X}}{\text{sup}} \left\{ \la C,{ X}\ra : X \in \mcs^n_+,\ \la A_i,{ X} \ra=b_i \ (i\in [m]) \right\}\tag{{P}}\label{primalintro},\\
&\underset{y,Z}{\text{inf}}\left\{ \sum_{i=1}^m b_iy_i\ :\ \sum_{i=1}^my_iA_i-C=Z\in \mcs^n_+\right\}\tag{{D}}\label{dualintro},
\end{align}
which we respectively denote by $(P)$ and $(D)$. Throughout, the set $\{1,2,\cdots, m \}$ is denoted by $[m].$ The sets 
\begin{equation*}
\begin{aligned} 
& \mcs^n_+ \cap \left\{ X\in \mcs^n: \la A_i,{ X} \ra=b_i \ (i\in [m]) \right\}, \\
 & \mcs^n_+\cap \left\{ \sum_{i=1}^my_iA_i-C: y\in \R^m\right\},
 \end{aligned}
 \end{equation*}
are the primal and dual feasible regions, respectively. Furthermore, a SDP is called {\em strictly feasible} if it has a full-rank (i.e., positive definite) feasible solution. 

The following result summarises SDP duality theory.

\begin{theorem}\label{sdpthm}
Consider a pair of primal dual SDPs $(P)$ and $(D)$.
The following properties hold:
\begin{itemize}

\item[$(i)$] (Complementary slackness) Let $X, (y,Z)$ be a pair of primal-dual feasible solutions for $(P)$ and $(D)$, respectively. Assuming that $p^*=d^*$ we have that $X, (y,Z)$ are primal-dual optimal if and only if $\la X, Z\ra=0.$
\item[$(ii)$] (Strong duality) Assume that $ d^*>-\infty$ (resp. $p^* <~+\infty$) and that (D) (resp. (P)) is strictly feasible. Then $p^*=d^*$ and furthermore, the primal (resp.\ dual) optimal value is attained. 
\end{itemize}
\end{theorem}
 
The following result gives a sufficient condition for showing that a SDP admits a unique optimal solution.

\begin{theorem}\cite{alizadeh}\label{uniqueness}
%Consider a pair of primal dual semidefinite programs as in 
Assume that the optimal values of~\eqref{primalintro} and~\eqref{dualintro} are equal and that both are attained. Let $Z$ be a dual optimal solution with the additional property that the homogeneous linear system 
\be\label{thm:uniqueness}
MZ=0, \ \la M,A_1\ra=0,\ldots, \la M,A_m\ra=0,
\ee
in the symmetric matrix variable $M$ only admits the trivial solution $M=0$. Then, the primal SDP \eqref{primalintro} has a unique optimal solution. 
\end{theorem} 

A dual optimal solution $Z$ for which \eqref{thm:uniqueness} only has  the  solution $M=0$ is called dual nondegenerate. Geometrically, condition \eqref{thm:uniqueness} expresses that the manifold of $n\times n$ matrices with rank equal to ${\rm rank}(Z)$ intersects the linear space ${\rm span}(A_1,\ldots,A_m)$ transversally at the point~$Z$.

Lastly, we introduce error bounds for SDPs, which allow to  bound  the distance of a feasible solution from the set  of optimal solutions, in terms of  the suboptimality of the objective function 
%, cf. Theorem \ref{}.
 (see \cite{sturm} or ~\cite{wolk} for a modern exposition). Given two sets $Q_1, Q_2$ lying in some Euclidean space, we say that they satisfy a {\em H\"older error bound} with exponent $q\ge 0$, if for any compact set $U$ we have that 
\be{\rm dist}(x,Q_1\cap Q_2)=\mathcal{O}({\rm dist}^q(x,Q_1)+{\rm dist}^q(x,Q_2)), \ \forall x\in U.
\ee
Recall that for a set $Q\subseteq \R^n$, the corresponding distance function is given by ${\rm dist}(x,Q)=\inf\{ \|x- q\|: q\in Q\}.$
Clearly, if $Q$ is closed, there always exists a point that achieves the minimum distance. If additionally $Q$ is convex, the distance is minimized at a unique point of $Q$. 
As it turns out, {\em feasibility SDPs} satisfy H\"olderian error bounds, where the  exponent depends neither on the size of the matrices $n$ nor the number of affine constraints~$m$.
%, see Theorem \ref{thm:errorboundsfeas}.
% in SM. CHECK

\begin{theorem}\label{thm:errorboundsfeas}
Consider the affine space 
$$\mathcal{A}=\{X\in \mcs^n:\ \la A_i,X\ra=b_i, \ (i\in [m])\},$$
where $A_1,\ldots,A_m\in \mathcal{S}^n$ and $b_1,\ldots,b_m\in \R.$ 
Then, for any compact set $U\subseteq \mcs^n$ there exists a constant $c=c(U)>0$ such that for all $X\in U$ we have that:
$${\rm dist}(X,\mathcal{A}\cap \mcs^n_+)\le c\left({\rm dist}^{2^{-d}}(X,\mathcal{A})+{\rm dist}^{2^{-d}}(X,\mcs^n_+)\right),$$
where $d\in \{0,\ldots,n-1\}$ is the {\em singularity degree} of \eqref{primalintro}, defined as the least number of facial reduction steps (as defined in \cite{borwein}),  required to make \eqref{primalintro} strictly~feasible. 
\end{theorem} 

Based on this, one can prove error bounds for general  (i.e., non-feasibility) SDPs: %see Theorem \ref{thm:errorbounds}.
%in SM. CHECK 

\begin{theorem}\label{thm:errorbounds} 
Consider a pair of primal/dual SDPs $(P)$ and $(D)$, where the primal/dual values are equal and both are attained. Furthermore, assume that the set of feasible solutions of $(P)$ is contained in some compact subset $U\subseteq~\mcs^n$. Setting $\mcp$ to be the set of primal optimal solutions and $d$ the {\em singularity degree} of $(P)$, we have that 
$${\rm dist}(\tilde{X}, \mcp)\le \mathcal{O}(\epsilon^{2^{-d}}),$$
for any primal feasible solution $\tilde{X}$ with 
$p^*-\epsilon \le \la C,\tilde{X}\ra$.

\end{theorem}

\begin{proof}
Let $Z^*=\sum_{i=1}^my^*_iA_i-C$ be a dual optimal solution. By Theorem \ref{sdpthm}, a matrix $X\in \mcs^n$ is primal optimal if and only if the following SDP is~feasible:
\be\label{optconditions}
 X\in \mcs^n_+,\ \la A_i,X\ra =b_i\ (i\in [m]),\ \la X,Z^*\ra=0.
\ee
In turn \eqref{optconditions} implies that $${\rm dist}(\tilde{X}, \mcp)={\rm dist}(\tilde{X}, \mathcal{A}\cap (Z^*)^\perp \cap\mcs^n_+).$$
 Thus, by Theorem \ref{thm:errorboundsfeas}, there exists a constant $c$ such that 
 $${\rm dist}(\tilde{X}, \mcp)\le c\left({\rm dist}^{2^{-d}}(X,\mathcal{A} \cap (Z^*)^\perp)+{\rm dist}^{2^{-d}}(X,\mcs^n_+)\right),$$
 for all $\tilde{X}\in U$. We proceed to bound the terms in the summand separately. Since $\tilde{X}\in \mcs^n_+$ we have that ${\rm dist}(\tilde{X},\mcs^n_+)=~0$. Furthermore, it follows by \cite{geometrybound} that there exists a constant $c'=c'(U)$ such that ${\rm dist}(\tilde{X},\mathcal{A} \cap (Z^*)^\perp)\le c'{\rm dist}(\tilde{X},(Z^*)^\perp).$ 
 It remains to upper bound $ {\rm dist}(\tilde{X},(Z^*)^\perp).$ The assumption $p^*-\epsilon \le \la C,\tilde{X}\ra$ implies that $\la Z^*,\tilde{X}\ra\le \epsilon$, since 
 $$p^*-\la C,\tilde{X}\ra=p^*-\la \sum_{i=1}^my^*_iA_i-Z^*,\tilde{X}\ra=
 %\sum_{i=1}^my^*_ib_i-p^*=
 \la Z^*,\tilde{X}\ra.$$
 Lastly, the distance of $\tilde{X}$ to $(Z^*)^\perp $ is equal to the length of the projection of $\tilde{X}$ onto ${\rm span}(Z^*)$, i.e.,
 $${\rm dist}(\tilde{X}, (Z^*)^\perp)={\la Z^*,\tilde{X}\ra\over \|Z^*\|}\le {\epsilon \over \|Z^*\|} .$$ 
\end{proof}
 %Note that Theorem \ref{thm:errorbounds} gives an upper for the distance of an arbitrary matrix to the set of feasible solutions of a SDP. Nevertheless, as optimality of an arbitrary SDP can be captured as a SDP feasibility problem (under very mild assumptions that are satisfied in our setting), an analogous statement also holds for nonfeasibility SDPs.
 
\section{Proof of the Main Theorem}\label{mainresultproof}

In this section we prove the Main Theorem. For the readers' convenience, we give  the proof in a series of~steps. 

\medskip 
\noindent {\bf Step 1a:}  The primal SDP \eqref{theta:primalmain} in the Main Text (MT) has a unique optimal solution, which we denote $X^*$ \cite{MT}.

\medskip
Indeed, by assumption  \eqref{algcon} from MT \cite{MT},  Theorem \ref{uniqueness} implies that  the primal SDP \eqref{theta:primalmain} (see MT \cite{MT}) has a unique optimal solution.

\medskip 
\noindent {\bf Step 1b:} The non-contextuality inequality is  a  {self-test}. %using the Lemma~\ref{STgramvec}.

\medskip
This follows from  the following Lemma:
\begin{lemma}\label{STgramvec}
Let $X^*$ be the unique optimal solution for the primal SDP $\eqref{theta:primalmain}$ (see MT \cite{MT}), and let $\{\ket{u_i}\bra{u_i}\}_{i=0}^n$ be a quantum realization achieving the maximum quantum value of $\sum_{i=1}^{n} w_ip_i  : p \in \mathcal{P}_{q}(\gexcl)$. Then, the noncontextuality inequality $\sum_{i=1}^{n} w_ip_i \le B_{nc}(\gexcl, w),$ for all $p \in \mathcal{P}_{nc}(\gexcl)$  is a \emph{self-test} for the realization $\{\ket{u_i}\bra{u_i}\}_{i=0}^n$.
\end{lemma}
\begin{proof}
By the equivalence between the optimization problems \eqref{cdfvdfb} (see MT \cite{MT}) and \eqref{theta:primalmain} (see MT \cite{MT}), as discussed in the main text, we have that 
 \be\label{SMequality}
 \begin{aligned}
 {{X^\star}}=&{\rm Gram}(\ket{u_0},\braket{u_0}{u_1}\ket{u_1},\ldots,\braket{u_0}{u_n}\ket{u_n})\\
=& {\rm Gram}(\ket{u_0'}, \braket{u_0'}{u_1'}\ket{u_1'},\ldots,\braket{u_0'}{u_n'}\ket{u_n'}), 
\end{aligned}
\ee
where $ \{\ket{u_i'}\bra{u_i'}\}_{i=0}^n$ is an arbitrary realization that also achieves the quantum value. As a consequence, there exists an isometry $V$ such that 
 \beq
 \label{isometry}
 V\ket{u_0} = \ket{u_0'}\ \text{ and }\ V\braket{u_0}{u_i}\ket{u_i}=\braket{u_0'}{u_i'}\ket{u_i'}\ (i\in [n]).
 \enq
 By equating the diagonal entries in \eqref{SMequality} we get that
\be|\braket{u_0}{u_i}|=|\braket{u_0'}{u_i'}|,
\ee
Thus for all $0 \leq i \leq n$,  we have that 
$$V\ket{u_i}\bra{u_i}V^{\dagger}=\ket{u_i'}\bra{u_i'},$$
which concludes  the proof.
\end{proof} 

\medskip 
\noindent {\bf Step 2:}  For any feasible solution $\tilde{X}$ satisfying 
$$ \sum_{i=1}^n w_i\tilde{X}_{ii}\ge B_{qc}(\gexcl, w)-\epsilon,$$ we have that 
\be \|\tilde{X}-X^*\|_F={\rm dist}(\tilde{X},X^*)\le \mathcal{O}(\epsilon). \ee

\medskip
We  prove this as a consequence of   Theorem \ref{thm:errorbounds}. For this,  we first show that  the set of feasible solutions of~\eqref{theta:primalmain} (see MT \cite{MT}) is bounded. Indeed, if $X$ is an arbitrary feasible solution, for any $i\in [n]$ we have that 
\be
\begin{pmatrix}1 & X_{0i} \\X_{0i} & X_{ii} \end{pmatrix}=\begin{pmatrix}1 & X_{ii} \\X_{ii} & X_{ii} \end{pmatrix}\succeq 0,
\ee
which in turn implies that $0\le X_{ii}\le 1$. Furthermore, for any $i\ne j\in [n]$, we have that 
\be\begin{pmatrix}X_{ii} & X_{ij} \\X_{ij} & X_{jj} \end{pmatrix}\succeq0,\ee
which implies that 
\be|X_{ij}|\le X_{ii}X_{jj}\le 1.
\ee
Next, we show that the SDP \eqref{theta:primalmain} (see MT \cite{MT}) is strictly feasible, (i.e., there exists a positive definite feasible solution), and consequently its singularity degree is equal to zero. Define $F_n = \frac{1}{m}I_{n+1} + \sum_{i=1}^n \frac{2}{m} E_{0i} +(1-\frac{1}{m}) E_{00} $,~i.e.,
\be
F_n =\left(\begin{array}{c|ccccc}
1 & \frac{1}{m} & \frac{1}{m} & \cdots & \cdots & \frac{1}{m} \\
\addlinespace[-10 pt] \\
\hline\\
\addlinespace[- 8pt]
\frac{1}{m} & \frac{1}{m} &0 & \cdots &\cdots & 0 \\
\frac{1}{m} &0 & \frac{1}{m} & \ddots & &\vdots \\
\vdots &\vdots & \ddots & \ddots & \ddots &\vdots \\
\vdots &\vdots & & \ddots & \ddots & 0 \\
\frac{1}{m} &0 & \cdots & \cdots & 0 & \frac{1}{m} \\
\end{array} \right),
\ee 
where $m > n$. 
It is easy to check that  $F_n$ satisfies all the feasibilty constraints and has full-rank.

 \medskip 
\noindent {\bf Step 3a:} 
For any quantum realization   $\{\ket{u'_i}\bra{u'_i}\}_{i=0}^n$  satisfying 
\be\label{csdfvrgbt}
\sum_{i=1}^n|\braket{u_i'}{u'_0}|^2\ge B_{qc}(\gexcl, w)-\epsilon,
\ee there exists a unitary matrix $U$ such that 
\be \|U\ket{u_i} - \ket{u'_i}\|_2 \leq \mathcal{O} (\sqrt{\epsilon}),\quad  0\le i\le n. \ee 

 This will follow from the following general fact concerning positive semidefinite matrices: If two positive semidefinite matrices are $\epsilon$-close in Frobenius distance, then their Gram decompositions are $\mathcal{O}\left(\sqrt{\epsilon}\right)$ close. Formally:

%Specifically, in Lemma~\ref{equivalentapprox} below, we  show that there exists an unitary matrix $U: \C^{(n+1)\times (n+1)}$ such that 
%\be \|U\ket{w_i} - \ket{w'_i}\|_2 \leq \mathcal{O}(\sqrt{\eps}),\ \text{ for all } i \in \{0,\cdots,n\},\ee where $\ket{w_0},\ldots,\ket{w_n}\in \C^{n+1}$ and $ \ket{w'_0},\ldots,\ket{w'_n}\in \C^{n+1}$ are Gram decompositions of $X^*$ and $\tilde{X}$, respectively. 

\begin{lemma}\label{equivalentapprox} 
Consider two positive semidefinite matrices $X,X'\in \mcs^n_+$ with $\|X-X'\|_F\le \epsilon$. For any two Gram decompositions $\ket{u_1},\ldots,\ket{u_n}\in \C^{n}$ and $ \ket{u'_1},\ldots,\ket{u'_n}\in \C^{n}$ of $X$ and $X'$, respectively, there exists a unitary matrix $U\in \C^{n\times n}$ such that 
\be
\|U\ket{u_i} - \ket{u'_i}\|_2 \leq \sqrt{n\eps},\ \text{ for all } i \in \{1,\ldots,n\}.
\ee
\end{lemma}

\begin{proof} 

%We make use of the following facts in the proof: 
%\begin{fact}\cite[Equation X.2, p.~290]{bhatia1997matrix}
%\label{continuity}
%Let $A,B \succeq 0.$ Then, \be  \|\sqrt{A}-\sqrt{B}\|_{\infty} \leq \sqrt{\|A-B\|_\infty}. \ee
%\end{fact}
%\begin{fact}\cite[Page 7]{bhatia1997matrix}
%\label{equivalence}
%Let $A \in \mathbb{C}^{n\times n}.$ Then,
%\be
%\|A\|_\infty \leq \|A\|_F \leq \sqrt{n}\|A\|_\infty.
%\ee
%\end{fact}

Note that $X=VV^\dagger$ and $X'=V'(V')^\dagger$, where $V\in~\C^{n\times n}$ has the vectors $\ket{u_1},\ldots,\ket{u_n}$ as its rows, and $V'\in \C^{n\times n}$ has the vectors $\ket{u'_1},\cdots,\ket{u'_n}$ as its rows. 

By the polar decomposition (\cite[Theorem~7.3.1]{HJ}) there exist unitary matrices $U_V\in \C^{n\times n}, U_{V'}\in \C^{n\times n}$ with 
\be\label{polar}
V= |V|U_V \ \text{ and }\ V'= |V'|U_{V'},
\ee
where we use the notation $|A|=(AA^\dagger)^{1/2}$. 

Furthermore, as $X=VV^\dagger$ and $X'=V'(V')^\dagger$ we have 
\be\label{sqroot}
\sqrt{X}=(VV^\dagger)^{1/2}=|V| \ \text{ and }\sqrt{X'}=(V'(V')^\dagger)^{1/2}=|V'|.
\ee
Combining \eqref{polar} and \eqref{sqroot}, we get 
\be\label{cvevbr}
V= \sqrt{X}U_V \ \text{ and } \ V'= \sqrt{X'}U_{V'}.
\ee
Using \eqref{cvevbr}, we get that 
\begin{align}
\|VU_V^\dagger U_{V'}-V'\|_F &=\|(\sqrt{X}-\sqrt{X'})U_{V'}\|_F \nonumber \\
&\leq \|\sqrt{X}-\sqrt{X'}\|_F\|U_{V'}\|_\infty. \enspace
\end{align} 

Since $U_{V'}$ is unitary, we have that $\|U_{V'}\|_\infty=1$. Lastly, the assumption $\|X-X'\|_F\le \epsilon$ implies that
\begin{align}
\|\sqrt{X}-\sqrt{X'}\|_F & \overset{a}\leq \sqrt{n}\|\sqrt{X} - \sqrt{X'}\|_{\infty} \overset{b}\leq \sqrt{n}\sqrt{\|X-X'\|_{\infty}} \nonumber \\
&\overset{c}\leq \sqrt{n}\sqrt{\|X-X'\|_{F}} \overset{d}\leq \sqrt{n\epsilon}, \enspace
\end{align}
where inequalities $(a)$ and $(c)$  follow from  the well-known relation $\|A\|_\infty \leq \|A\|_F \leq \sqrt{n}\|A\|_\infty$ (see \cite[Page 7]{bhatia1997matrix}),  inequality $(b)$ follows from   $\|\sqrt{A}-\sqrt{B}\|_{\infty} \leq \sqrt{\|A-B\|_\infty}$ for $ A,B \succeq 0$ (see \cite[Equation X.2, p.~290]{bhatia1997matrix}), and $(d)$ follows from the assumption. 
 
Summarizing, we have shown that 
\be \|VU-V'\|_F\le \sqrt{n\epsilon},
\ee
for the unitary matrix $U=U^\dagger_{V'}U_V\in \C^{n\times n}$. 
\end{proof} 
Now, we can conclude the proof of Step 4. Set  $$\ket{w_0} = \ket{u_0}, \quad \ket{w_i} = \braket{u_0}{u_i}\ket{u_i},$$ 
$$\ket{w'_0} = \ket{u'_0}, \quad \ket{w'_i} = \braket{u'_0}{u'_i}\ket{u'_i},$$
and let $X$ and $\tilde{X}$ be the Gram matrices of $\{\ket{w_i}\}_{i=0}^n$ and $\{\ket{w'_i}\}_{i=0}^n$ respectively. By \eqref{csdfvrgbt} we have that 
$$ \sum_{i=1}^n w_i\tilde{X}_{ii}\ge B_{qc}(\gexcl, w)-\epsilon,$$ and thus, it follows  by Step 3  that 
\be\label{vdfbfbnt}
 \|\tilde{X}-X^*\|_F\le \mathcal{O}(\epsilon). \ee
In turn, \eqref{vdfbfbnt}  combined with  Lemma \ref{equivalentapprox} imply that  there exists a unitary matrix $U\in \C^{n\times n}$ such that 
\be
\|U\ket{w_i} - \ket{w'_i}\|_2 \leq \sqrt{n\eps},\ \text{ for all } i \in \{1,\ldots,n\}.
\ee
Lastly,  noting that 
\be
 \|\ket{w_i}\|_2=|\braket{u_0}{u_i}| = \min_{1\le i\le n} |\braket{u_0}{u_i}|>0,%\sqrt{\frac{\vartheta_n}{n}},
\ee where the strict positivity follows from the assumption    $\braket{u_0}{u_i}\ne 0, $ for all $1\le i\le n$. Thus, we have that 
\be \|U\ket{u_i} - \ket{u'_i}\|_2 \leq \mathcal{O} (\sqrt{\epsilon}),\quad  0\le i\le n, \ee
as a consequence of the following simple lemma:

\begin{lemma} \label{unitvecdist}
Let $\|\ket{a} - \ket{b} \|_2 \leq \delta$ for vectors $\ket{a}$ and $\ket{b}$, such that $\| \ket{a}\|_2 \geq 2\delta$. Denote by $\ket{\hat{a}}$ and $\ket{\hat{b}}$ to  be $\ket{a}$ and $\ket{b}$, normalised to have unit norm. 
%the unit vectors of $\ket{a}$ and $\ket{b}$, respectively. 
Then, we have that $\| \ket{\hat{a}}- \ket{\hat{b}} \|_2 \leq \frac{2\delta}{\|\ket{a}\|_2}$.
\end{lemma}

\begin{proof}To ease notation  we use  vector notations instead of the bra-ket notation and  set $\|\cdot\|=\|\cdot\|_2$ throughout this proof. 
Suppose $\|a\| \leq \|b\|$. By assumption, we have  
\be \left\| \|a\|\hat{a} - \|b\|\hat{b} \right\| \leq \delta. \ee
Now observe that $\| \hat{a} -\hat{b} \|^2 = 2-2|\hat{a}^{\top}\hat{b}|$, whereas 
\be \| \hat{a} - \frac{\|b\|}{\|a\|}\hat{b} \|^2 = 1 + \left(\frac{\|b\|}{\|a\|}\right)^2 - 2\frac{\|b\|}{\|a\|}|\hat{a}^{\top}\hat{b}|.\ee From the assumption that $\|a\| \leq \|b\|$, we have $\| \hat{a} - \frac{\|b\|}{\|a\|}\hat{b} \|^2 - \| \hat{a} - \hat{b} \|^2 \geq 0$, and hence $\| \hat{a} - \hat{b} \| \leq \frac{\delta}{\|a\|}$. The same arguments hold when $\|a\| \geq \|b\|$, by switching $a$ and $b$. Hence, in general, we have $\| \hat{a}- \hat{b} \|\leq \frac{\delta}{\min \{\|a\|,\|b\| \}}$.

Note that $\|a\| - \|b\| \leq \|a - b \| \leq \delta$, implies that $\|b\| \geq \|a\| - \delta$ and therefore $\min \{\|a\|,\|b\| \} \geq \|a\| - \delta$. So, finally, we have 
\be \| \hat{a}- \hat{b} \| \leq \frac{\delta}{\|a\| - \delta} \leq \frac{\delta}{\|a\|}\left(1+ \frac{2\delta}{\|a\|}\right) \leq \frac{2\delta}{\|a\|},\ee where for the second inequality we use the fact that $\frac{1}{1-x} \leq 1+ 2x,$ for $0 \leq x \leq \frac{1}{2}$.
\end{proof}

 \medskip 
\noindent {\bf Step 3b:} We have that 
 \be \|U\ket{{u_i}}\bra{{u_i}}U^{\dagger} - \ket{{u'_i}}\bra{{u'_i}}\| \leq \mathcal{O}\left({\sqrt{\epsilon}}\right), \quad  0\le i\le n.\ee

\medskip

\begin{lemma}\label{vec_proj}
For two unit vectors $\ket{x}$ and $\ket{y}$, if $\|\ket{x}-\ket{y}\|_2 < \epsilon$, then $\| \ket{x}\bra{x} -\ket{y}\bra{y}\|_F < \sqrt{2} \epsilon$. 
\end{lemma}

\begin{proof}
Since $\|\ket{x}-\ket{y}\|_2 < \epsilon$, we have 
\begin{align*}
\epsilon^2 &> \|\ket{x}-\ket{y}\|_2^2 = \sum_i (x_i -y_i)^2 = \sum_i x_i^2 + y_i^2 - 2x_iy_i \\ &= 2-2 \sum_i x_iy_i,
\end{align*}
implying, 
\begin{equation}\label{vecprojdist}
\sum_i x_iy_i > 1- \frac{\epsilon^2}{2}.
\end{equation}

Now we bound the square  distance between the corresponding projectors. 

\begin{align*}
\| xx^{T} - yy^{T}\|_F^2 &= \sum_{i,j} (x_ix_j - y_iy_j)^2  \\
&= \sum_{i,j} x_i^2x_j^2 + y_i^2y_j^2 - 2x_ix_jy_iy_j \\
&= 2 - 2 \left(\sum_i x_iy_i \right)^2 \\
& < 2-2\left(1- \frac{\epsilon^2}{2} \right)^2 < 2\epsilon^2, 
\end{align*}
where the second last inequality follows from $\eqref{vecprojdist}$. Taking square roots on both sides, finishes the proof. 
\end{proof}

 \section{Proof of Uniqueness for KCBS}\label{oddcycleunique}

In this section we prove Corollary~\ref{maincorollary} (see MT \cite{MT}). The first observation is that the given quantum realization $\{\ket{u_i}\bra{u_i}\}_{i=0}^n$, achieves the maximum value of the $KCBS_n$ inequality. In order to prove that for any odd integer $n$,  the $KCBS_n$ inequality  is an  $(\epsilon,{1\over 2})$-robust self-test for this realization, it suffices to show that the the primal SDP~\eqref{theta:primalcn}, has a unique optimal solution. Equivalently, we show that the dual optimal solution $Z^*$ for \ref{dualcnintro} (see MT~\cite{MT}) corresponding to an odd $n$-cycle graph, satisfies the condition given in \ref{algcon} (see MT \cite{MT}). We prove this in Theorem~\ref{wherethemagichappens}.

\begin{theorem}\label{wherethemagichappens}
For any odd integer $n$, the SDP 
\be\label{theta:primalcn}
\begin{aligned} 
\vartheta(C_n) = \max & \ \sum_{i=1}^n { X}_{ii} \\
{\rm s.t.} & \ { X}_{ii}={ X}_{0i}, \  i\in [n],\\
 & \ { X}_{i(i+1)}=0,\ i\in [n],\\
& \ X_{00}=1,\ X\in \mathcal{S}^{1+n}_+,
\end{aligned}
\ee
satisfies the property that the matrix
\be X^*={\rm Gram}\left(\ket{u_0},\braket{u_0}{u_1}\ket{u_1},\ldots,\braket{u_0}{u_n}\ket{u_n}\right),\ee
is the unique optimal solution of \eqref{theta:primalcn}, where $ \{\ket{u_i}\bra{u_i}\}_{i=0}^n$ is the canonical realization given in the Corollary~\ref{maincorollary} in MT. 
\end{theorem} 

\begin{proof} By Theorem \ref{uniqueness}, the SDP \eqref{theta:primalcn} has a unique optimal solution if there exists a dual optimal solution $Z$ for which, the homogeneous linear system 
\be
M_{00}=0, \ M_{0i}=M_{ii}, \ M_{i(i+1)}=0\ (i\in [n]), \ MZ = 0,
\ee
in the symmetric matrix variable $M$, only has the trivial solution $M=0$. 
The dual of \eqref{theta:primalcn} is equal to the least $t\ge 0$ for which 
\be \label{theta:dualcn}
 tE_{00}+\sum_{i=1}^n E_{ii}(\lambda_i-w_i)-\sum_{i=1}^n\lambda_i E_{0i}+\sum_{i=1}^n \mu_{i(i+1)}E_{i(i+1)}\succeq 0.
\ee
We show that for any odd positive integer $n,$ the following matrix is an optimal solution for \eqref{theta:dualcn}
 \be\label{zmatrix}
{ Z}_n = \left[ 
\begin{array} {c | c}
\vartheta(C_n) & -e^\top\\
\hline 
-e & I_n+\frac{n -\vartheta(C_n)}{2 \vartheta(C_n)}A_{C_n}
\end{array}
\right]\in \R^{(1+n)\times (1+n)},
\ee
where $e$ is the all one vector of size $n$ and $A_{C_n}$ is the adjacency matrix of $C_n$. 
By construction, the value of $Z_n$ is equal to $\vartheta(C_n)$. It remains to show feasibility. Note that $Z_n$ is an admissible solution to \eqref{theta:dualcn}, where $t=\vartheta(C_n)$, $\lambda_i=-1$, $\nu_i=0$ and $\mu_i=(n -\vartheta(C_n))/(2 \vartheta(C_n))$. Thus, it remains to prove that $Z_n$ is positive semidefinite, i.e., it has nonnegative eigenvalues. 
Taking the Schur complement of $Z_n$ with respect to its top left entry, we have that 
\be\label{scomplemet}
Z_n\succeq 0 \Longleftrightarrow I_n+\frac{n -\vartheta(C_n)}{2 \vartheta(C_n)}A_{C_n}-{1\over \vartheta(C_n)}ee^\top\succeq 0.
\ee
 As the graph $C_n$ is 2-regular, the all ones vector $e$ is an eigenvector of $A_{C_n}$ with corresponding eigenvalue $2$. Thus, the matrices $A_{C_n}$ and $ee^\top$ commute. Consequently, the eigenvalues of the matrix $I_n+\frac{n -\vartheta(C_n)}{2 \vartheta(C_n)}A_{C_n}-{1\over \vartheta(C_n)}ee^\top$~are~$0$  and 
 \be\label{eigs}
 1+\frac{n -\vartheta(C_n)}{\vartheta(C_n)}\cos {2\pi k\over n},\ 1\le k \le n-1,
 \ee where we used the well-known fact that the eigenvalues of $A_{C_n}$ are given by $\omega^k+\omega^{-k}$, for $k=0,\ldots, n-1$, where $\omega=\exp (2\pi i /n)$ is an $n$-th root of unity. 

To prove that $Z_n$ is positive semidefinite, it remains to show that the smallest eigenvalue from \eqref{eigs} is nonnegative. Note that the smallest eigenvalue is obtained for $k=(n-1)/2$, and is equal to
\be1-\frac{n -\vartheta(C_n)}{\vartheta(C_n)}\cos {\pi \over n}=0,
\ee
where the last equality follows from $\vartheta(C_n)=\frac{ n \cos \pi/ n}{ 1+ \cos\pi / n}$. 

Lastly, we show 
that any odd $n,$ the only symmetric matrix $M \in \mathbb{R}^{(1+n) \times (1+n)} $ satisfying 
\be 
M_{00}=0, \ M_{0i}=M_{ii}, \ M_{i(i+1)}=0\ (i\in [n]), \ MZ_n=0,
\ee
is the matrix $M=0$. 

 For notational convenience, let 
$ MZ_n=\left(\begin{smallmatrix} 
{q}& {r^\top}\\ 
{r} & {T}
\end{smallmatrix}\right)$. 
%where $ Q $ is a $1 \times 1$ table, $ R $ is a $1 \times n$ table, $ S $ is a $n \times 1$ table and $ T $ is a $n \times n$ table of linear equations.
 In the rest of this section we use the notation $\underline{i}$ to denote $i \mod n,$ where $i$ is an integer. Also, we denote by $A$ the $n\times n$ bottom-right submatrix of $M$.  The linear equation corresponding to $q$ implies that
 \beq
\label{trace}
 \Tr \left( M \right)= 0.
\enq 
For $i \in \{ 1,2,\cdots, n \}$, the $2n$ linear equations corresponding to ${ T}_{i,\underline{i+1}}$ and ${ T}_{i,\underline{i-1}}$ imply 
\beq
\label{equalA}
{{ A}_{i,\underline{i+2}}} = {{ A}_{i,\underline{i-2}}} = \alpha {{ A}_{i,i}},
\enq
where $\alpha = \left(2 \cos\left(\frac{\pi}{n}\right)-1\right).$ Since $ A$ is symmetric,
\beq
\label{equalB}
{{ A}_{\underline{i+2},i}} = {{ A}_{i,\underline{i+2}}} =\alpha {{ A}_{\underline{i+2}, \underline{i+2}}}
\enq
and
\beq
%\label{equa}
{{ A}_{\underline{i-2},i}} = {{ A}_{i,\underline{i-2}}} = \alpha {{ A}_{\underline{i-2}, \underline{i-2}}},
\enq
for $i \in \{ 1,2,\cdots, n \}$.
Using \eqref{equalA} and \eqref{equalB}, we get 
\beq
\label{equalC}
{{ A}_{i,i}} = {{ A}_{j,j}},
\enq
for $i, j \in \{ 1,2,\cdots, n \}$. Finally, using \eqref{trace}, \eqref{equalA}, \eqref{equalB}, and \eqref{equalC}, we get
\beq
\label{equalD}
{{ A}_{i,i}} = {{ A}_{i,\underline{i+2}}}= {{ A}_{i,\underline{i-2}}} = 0.
\enq
Now, using ${ T}_{i,\underline{i+k}}$, we infer 
\beq
\label{equalE}
{{ A}_{i,\underline{i+k+1} }} = 0,
\enq
for $k \in \{2,\cdots,n-4\}.$ Thus ${ A} = 0$ and since ${{A}}_{i,i} = {{K}}_{1,i} $, this finally leads to ${M} = 0.$
 
\end{proof} 

 \section{Example: 5-Cycle}\label{five-cycle}
 Now we explain our procedure for proving self testing for the special case of 5 cycle. The optimal solution for the dual \ref{dualcnintro} (see MT \cite{MT}) for $n=5,$ is given by

%Solving the dual for 5-cycle we get the following dual optimal solution: 
%$$t=\sqrt{5}, \quad \mu_i=-2, \quad \lambda_{ij}={5-\sqrt{5}\over 2\sqrt{5}}, i.e., $$
\be\label{dualsolution}
Z_5=\left(\begin{array}{c|ccccc} 
\sqrt{5} & -1 & -1 & -1 & -1 & - 1\\
\hline
-1 & 1 & c& 0 & 0 & c\\
-1 & c & 1 & c& 0 & 0 \\
-1 & 0 & c & 1 & c & 0\\
-1 & 0 & 0 & c & 1 & c\\
-1& c & 0 & 0& c & 1
\end{array}\right), 
\ee
where $c={5-\sqrt{5}\over 2\sqrt{5}}.$ 
The $M$ matrix introduced in Theorem~\ref{uniqueness} for this example is given by
\be\label{x_5}
M_5=\left(\begin{array}{c|ccccc} 
0 & m_1 & m_2 & m_3 & m_4 & m_5\\
\hline
m_1 & m_1  & 0&  m_6& m_9 & 0\\
m_2 & 0 & m_2  & 0& m_7 & m_{10} \\
m_3 & m_6& 0 & m_3  & 0& m_8\\
m_4 & m_9 & m_7 & 0& m_4 & 0\\
m_5 & 0 & m_{10} & m_8& 0 & m_5 
\end{array}\right), 
\ee
To show that $Z_5$ is dual-nondegenerate  optimal solution to  \eqref{theta:primalmain} (see MT \cite{MT}), we need to show that $M_5 = 0 $ is only solution to the system of linear equations $M_5 Z_5 = 0$. Using equations \eqref{equalA}to \eqref{equalE} for $n = 5$, $M_5=0$ turns out to be the only solution to $M_5 Z_5 = 0$. Thus optimal solution to dual is non-degenerate. Using Theorem \ref{uniqueness}, the primal optimal solution for \eqref{theta:primalmain} is unique. Using Main Theorem, it is evident that KCBS inequality corresponding to the $5$ cycle scenario admits robust self-testing.
The optimal solution for primal is given by:
\be\label{primalsolution}
X_5^{\star}=\left(\begin{array}{c|ccccc} 
1 & d & d  & d & d &d \\
\hline
d & d & 0& f & f & 0\\
d & 0 & d& 0& f & f \\
d & f & 0& d & 0 & f\\
d & f & f & 0& d & 0\\
d& 0 & f & f& 0 & d
\end{array}\right), 
\ee
where $d={1 \over \sqrt{5}}$ and $f =\beta + (1- \beta)\cos(0.4 \pi) $ such that $\beta =\frac{ \cos\left(\pi \over 5\right)}{ 1+ \cos\left(\pi \over 5\right)}.$ The Gram decomposition of \eqref{primalsolution} gives the optimum measurement settings and the state. Note that the matrix $X_5^{\star}$ is rank 3, and hence gives a gram decomposition over qutrits. One such canonical decomposition is given by:
$ \ket{u_0} = \left(1,0,0\right)$ and
 $\ket{u_j} = \left( \cos(\theta), \sin(\theta)\sin\left( \phi_j\right), \sin(\theta)\cos\left(\phi_j\right) \right),$
%\ket{v_j} &= \left( \cos(\theta), \sin(\theta)\sin\left(\frac{j \pi (n-1)}{n}\right), \sin(\theta)\cos\left(\frac{j \pi (n-1)}
%{n}\right) \right)^\top,
%\end{aligned}
%\ee
where 
 $\cos^2(\theta)= \frac{ \cos\left(\pi/ n\right)}{ 1+ \cos\left(\pi / n\right)} \text{ and } \phi_j = \frac{j \pi (n-1)}{n} \text{ for } 1\le j\le n.$

\end{document}